\documentclass[letterpaper,10 pt,conference]{ieeeconf}
\IEEEoverridecommandlockouts 
\usepackage{amsmath,amssymb,amsfonts}
\usepackage{graphicx}
\usepackage{subfig}
\usepackage{textcomp}
\usepackage[T1]{fontenc}
\usepackage{color}
\usepackage{xcolor}
\usepackage{multirow}
\usepackage{multicol,lipsum}
\usepackage{epsfig}
\usepackage{algorithm}
\usepackage{algorithmic}
\usepackage{multirow}
\usepackage{hyperref}
\usepackage{mathtools}
\usepackage{bm}
\usepackage{listings}

\newtheorem{definition}{Definition}
\newtheorem{theorem}{Theorem}
\newtheorem{problem}{Problem}
\newtheorem{proposition}{Proposition}
\newtheorem{lemma}{Lemma}
\newtheorem{example}{Example}

\usepackage{array}
\usepackage{float}
\usepackage[short]{optidef}
\usepackage{mdframed}
\mdfdefinestyle{mystyle}{
    backgroundcolor=white!10}

\newcommand{\oomit}[1]{}

\newcolumntype{M}[1]{>{\centering\arraybackslash}m{#1}}
\newcolumntype{N}{@{}m{0pt}@{}}

\begin{document}

\title{Quantitative Verification of Finite-Time Constrained Occupation Measures for Continuous-time Stochastic Systems}

\author{Bai Xue$^1$ and C.-H. Luke Ong$^2$\\
\small 1. KLSS, Institute of Software, Chinese Academy of Sciences, Beijing, China\\
2. College of Computing and Data Science, Nanyang Technological University, Singapore\\
Email: xuebai@ios.ac.cn; luke.ong@ntu.edu.sg
} 

\maketitle
\thispagestyle{empty}
\pagestyle{empty}

\begin{abstract}
This paper addresses the quantitative verification of finite-time constrained occupation time for stochastic continuous-time systems governed by stochastic differential equations (SDEs). Unlike classical reachability analysis, which focuses on single-event properties such as entering a target set, many autonomous tasks—including surveillance, wireless charging, and chemical mixing—require a system to accumulate a prescribed duration within a target region while strictly maintaining safety constraints. We propose a barrier-certificate framework to compute rigorous upper and lower bounds on the probability that such cumulative specifications are satisfied over a finite time horizon. By introducing a stopped process that freezes the system once it reaches the boundary of the safe set, we derive three classes of certificates: one for upper bounds and two for lower bounds. The proposed approaches are validated through numerical examples implemented using semidefinite programming.
\end{abstract}

\section{Introduction}

Stochastic continuous-time systems serve as the canonical mathematical model for safety-critical processes operating under uncertainty, ranging from autonomous aerial vehicles subject to wind gusts to chemical reactors influenced by thermal noise \cite{oksendal2003stochastic}. Formal verification of these systems has traditionally prioritized binary properties: \textit{safety} (ensuring the system never enters an unsafe state) \cite{Prajna2007} or \textit{reachability} (ensuring the system eventually hits a target state) \cite{Xue2023,Xue2024Framework,jafarpour2025probabilistic}. However, for a broad class of autonomous applications, these binary metrics are insufficient to capture the desired system performance.

Consider an autonomous underwater vehicle (AUV) inspecting a submerged pipeline. Due to strong and stochastic ocean currents, stationary hovering is infeasible. Instead, the AUV must execute repeated passes or dynamic station-keeping maneuvers to keep its sensors focused on the target area. The mission requires accumulating a prescribed amount of valid scan data—corresponding to occupation time—while counteracting drift forces that may push the vehicle toward hazardous obstacles. In such scenarios, the appropriate metric of success is the constrained occupation time: the total time spent in a target set before any safety violation occurs. This notion is strictly stronger than the classical occupation time \cite{darling1957occupation}, which accounts only for the dwell time in the target and does not incorporate safety constraints.

In this paper, we develop a barrier-certificate framework for the quantitative verification of finite-time constrained occupation properties. Given a finite time horizon and a required service duration, our objective is to compute rigorous upper and lower bounds on the constrained occupation probability. Specifically, we quantify the probability that the system spends at least the prescribed amount of time within a compact target region during the horizon while remaining inside a designated open, bounded safe set throughout the process. To address this problem, we extend stochastic barrier functions—traditionally used for safety and reachability analysis—to capture the cumulative “memory” of time spent in the target set. Within this unified framework, we introduce three complementary classes of barrier certificates to derive rigorous probability bounds:
\begin{itemize}
    \item \textbf{Dissipative Barriers (Upper Bounds)}: These derive upper probability bounds by imposing a dissipative balance between exponential decay and additive drift.
\item \textbf{Attractive Barriers I (Lower Bounds)}: These yield lower bounds by penalizing time spent outside the target set while relying on a restoring tendency that drives trajectories back toward the target. In this class the offset satisfies $\beta\le 0$: the case $\beta=0$ typically corresponds to dynamics that are effectively attractive toward the target, whereas $\beta<0$ models leaky or transient regimes in which the required occupation time must be accumulated through repeated visits.

\item \textbf{Attractive Barriers II (Lower Bounds)}: These are intended for systems with strict attraction and employ a bidirectional weighted formulation (the two-speed clock). The construction exploits a positive drift term ($\beta>0$) favoring the target, yielding tighter bounds for strongly attracting systems.
\end{itemize}
The effectiveness of the proposed framework is demonstrated on polynomial stochastic systems, where the barrier conditions are formulated as semidefinite programs (SDPs) and solved using sum-of-squares (SOS) techniques.

\subsection*{Related Work}
The verification of stochastic systems has become a powerful paradigm, with barrier certificates serving as an effective tool for deriving rigorous probability bounds \cite{Prajna2007,Xue2021,Zikelic2023,Xue2023,Xue2024,Prajna2007}. These techniques have been extended to probabilistic programs and $\omega$-regular properties~\cite{chakarov2013probabilistic,abate2025quantitative,henzinger2025supermartingale,wang2025verifying}; see \cite{cao2025comparative} for a survey. However, most existing work focuses on infinite-horizon guarantees, whereas practical systems often require finite-time assurances due to operational constraints. This motivates the study of finite-time analysis, which we review next.

\noindent\textbf{Discrete-Time Systems.} 
$c$-martingale-based barrier conditions were proposed for formal verification for stochastic discrete-time systems, with foundational contributions in~\cite{kushner1967,steinhardt2012,Jagtap2018,Jagtap2020,Santoyo2021}. They allow the expected value of a barrier function to increase over time, naturally yielding lower bounds on safety probabilities, though they do not provide upper bounds. More recently,~\cite{Zhi2024,Xue2024Finite} proposed barrier conditions capable of deriving both lower and upper bounds on finite-time safety and reach-avoid probabilities.

\noindent\textbf{Continuous-Time Systems.} 
Early work on finite-time verification for continuous-time systems modeled by SDEs leveraged comparison theorems to bound exit probabilities using one-dimensional reference processes, such as the Ornstein–Uhlenbeck process~\cite{Nilsson2020}. The concept of $c$-martingales was adapted to continuous time in~\cite{steinhardt2012} and subsequently refined in~\cite{Santoyo2021} through state-dependent bounds on certificate growth to tighten verification results. Later,~\cite{Feng2020} introduced time-varying barrier functions satisfying Doob's supermartingale inequality over finite intervals. Recent frameworks have moved beyond strict nonnegative supermartingales by relaxing pratial differential equations governing finite-time reachability or applying Gr\"onwall's inequality to the generator's differential form, leading to new barrier conditions in~\cite{Xue2024CDC,Xue2024Framework} that yield both lower and upper bounds on safety and reach-avoid probabilities. 

\noindent\textbf{Differences from Current Work.} The aforementioned works primarily address single-event properties, such as safety (remaining within a set) or reachability (hitting a target set at least once). In contrast, this paper addresses the problem of constrained occupation time—ensuring a system accumulates a specific duration in a target region while remaining safe. This represents a significant generalization of the reach-avoid problem to cumulative behaviors. 

\noindent\textbf{Organization.} Section \ref{sec:preliminaries} introduces the problem formulation and mathematical preliminaries. Section \ref{sec:main_results} details our barrier conditions for upper and lower bounding finite-time constrained occupation probabilities. Then, we demonstrate the efficacy of our approach through numerical examples in Section \ref{sec:examples} and conclude in Section \ref{sec:conclusion}.

\section{Preliminaries}
\label{sec:preliminaries}

This section defines the stochastic system model within a formal probability space, formulates the occupation time metric, and reviews the necessary stochastic calculus tools.

\noindent\textbf{Notations:} $\mathbb{R}$ is the set of real numbers; $\Delta_1 \setminus \Delta_2$ denotes set difference, $\overline{\Delta_1}$ the closure of $\Delta_1$, and $\partial \Delta_1$ its boundary; and $1_A(\bm{x})$ is the indicator function of a set $A$; $1_A$ denotes whether the event $A$ is true.

\subsection{Problem Statement}

We consider a continuous-time stochastic dynamical system defined on a probability space $(\Omega, \mathcal{F}, \mathbb{P})$, where $\Omega$ is the sample space, $\mathcal{F}$ is the $\sigma$-algebra of events, and $\mathbb{P}$ is the probability measure. The space is endowed with a natural filtration $\mathbb{F} = (\mathcal{F}_t)_{t \geq 0}$, which represents the accumulation of information over time. The evolution of the system state $X_t \in \mathbb{R}^n$ is governed by the It\^o SDE:
\begin{equation}
    dX_t = f(X_t)dt + \sigma(X_t)dW_t, \quad X_0 = x_0 \in \mathcal{X},
    \label{eq:sde}
\end{equation}
where $f: \mathbb{R}^n \to \mathbb{R}^n$ is the drift vector field, $\sigma: \mathbb{R}^n \to \mathbb{R}^{n \times m}$ is the diffusion matrix, and $W_t$ is a standard $m$-dimensional Brownian motion adapted to the filtration $\mathcal{F}$. We assume $f$ and $\sigma$ are locally Lipschitz continuous to ensure the existence and uniqueness of solutions.

We analyze the system's behavior with respect to two specific subsets of the state space:
\begin{enumerate}
    \item \textbf{Safe Set ($\mathcal{X} \subset \mathbb{R}^n$):} An open, bounded domain representing the valid operating region. 
    \item \textbf{Target Set ($\mathcal{T} \subseteq \mathcal{X}$):} A measurable, compact subset of the safe set representing the region of interest where service time is accumulated.
\end{enumerate}

\noindent\textbf{Safety Exit Time.} The safety exit time $\tau_{safe}$ is the first instant the system leaves the safe set:
\begin{equation*}
    \tau_{safe} = \inf \{ t > 0 : X_t \notin \mathcal{X} \}.
\end{equation*}
If the trajectory remains in $\mathcal{X}$ indefinitely, we set $\tau_{safe} = \infty$.

\noindent\textbf{Constrained Occupation Time.} Unlike classical reachability, which concerns a single visit, we focus on the cumulative duration the system resides in $\mathcal{T}$ while strictly maintaining safety. The constrained occupation time $O_\mathcal{T}(t)$ over the interval $[0, t]$ is defined as:
\begin{equation}
\label{opt}
    O_\mathcal{T}(t) = \int_{0}^{t \land \tau_{safe}} 1_\mathcal{T}(X_s) \, ds.
\end{equation}
 Crucially, if the system violates safety ($\tau_{safe} < t$), the accumulation of occupation time ceases at $\tau_{safe}$.

We address the quantitative verification of finite-time occupation measures.
\begin{problem} 
\label{problem}
Given a horizon $H \in (0, \infty)$ and a threshold $K \in (0, H]$, compute upper and lower bounds on the constrained occupation probability $\mathbb{P}(O_\mathcal{T}(H) \geq K)$. 
\end{problem}

\subsection{Infinitesimal Generator}

The connection between the stochastic dynamics and our barrier functions is established via the infinitesimal generator $\mathcal{L}$. For a twice continuously differentiable function $v \in C^2(\mathbb{R}^n)$, the operator associated with the SDE \eqref{eq:sde} is defined:
\begin{equation}
    \mathcal{L}v(x) = \nabla v(x) \cdot f(x) + \frac{1}{2} \operatorname{Tr}\left( \sigma(x)^\top H_v(x) \sigma(x) \right),
    \label{eq:generator}
\end{equation}
where $\nabla v$ is the gradient, $H_v$ the Hessian matrix of $v$, and $\operatorname{Tr}(\cdot)$ denotes the trace operator, i.e., the sum of the diagonal entries of a square matrix.

A fundamental tool connecting the generator to probabilistic quantities is Dynkin's formula \cite{oksendal2003stochastic}. 

\begin{definition}[Stopping Time]
A random variable $\tau: \Omega \to [0,\infty)$ is a \textit{stopping time} with respect to the filtration $\{\mathcal{F}_t\}_{t \ge 0}$ if, for every deterministic time $t \ge 0$, the event $\{\tau \leq t\}$ is $\mathcal{F}_t$-measurable.
\end{definition}

For any bounded stopping time $\tau$ and any function $v \in C^2(\mathbb{R}^n)$ with compact support (or satisfying suitable polynomial growth conditions), we have:
\begin{equation*}
    \mathbb{E}[v(X_\tau)\mid X_0=x_0] = v(x_0) + \mathbb{E}\left[ \int_0^\tau \mathcal{L}v(X_s) \, ds \mid X_0=x_0\right].
\end{equation*}
This formula plays a central role in deriving the barrier certificate conditions in the following sections.

\subsection{Martingales and Markov's Inequality}
This subsection recalls fundamental results from martingale theory \cite{williams1991probability}—specifically martingales and Markov's inequality—that enable the rigorous analysis of system behavior over random horizons.

\begin{definition}[Continuous-Time Martingale]
Let $(\Omega,\mathcal{F},\mathbb{P})$ be a probability space and 
$\{\mathcal{F}_t\}_{t\ge0}$ be a filtration satisfying the usual
conditions (right-continuous and complete). 
An $\{\mathcal{F}_t\}$-adapted stochastic process $\{X_t\}_{t\ge0}$ 
is called a \emph{martingale} if the following conditions hold:
\begin{enumerate}
    \item $\mathbb{E}[|X_t|] < \infty$ for all $t \ge 0$;
    \item For all $0 \le s \le t$, $\mathbb{E}[X_t \mid \mathcal{F}_s] = X_s, \quad \text{almost surely}$.
\end{enumerate}
If the equality is replaced by 
$\mathbb{E}[X_t \mid \mathcal{F}_s] \le X_s$, the process is called a 
\emph{supermartingale}. If it is replaced by 
$\mathbb{E}[X_t \mid \mathcal{F}_s] \ge X_s$, the process is called a 
\emph{submartingale}.
\end{definition}

To handle analysis over random horizons (such as $\tau_{\text{safe}}$), we utilize two fundamental results from martingale theory.

\begin{proposition}[Stochastic integral is a martingale]
\label{pro:martin}
Let $W_t$ be a standard Brownian motion adapted to $\{\mathcal{F}_t\}$.
Let $H_t$ be an $\{\mathcal{F}_t\}$-predictable process such that for every $t\ge0$, $\mathbb{E} \left[\int_0^t |H_s|^2 \,ds\right] < \infty$. Then the stochastic integral $M_t := \int_0^t H_s \, dW_s$ is a (square-integrable) martingale with $\mathbb{E}[M_t]=0$ for all $t\ge0$.
Moreover, if $\tau$ is any bounded stopping time (i.e. $\tau\le H$ a.s.\ for some constant $H$),
then $M_{t\wedge\tau}$ is a martingale and in particular $\mathbb{E}[M_\tau]=0$.
\end{proposition}

\begin{proposition}[Markov's Inequality]
\label{thm:markov}
Let $X$ be a non-negative random variable defined on a probability space $(\Omega, \mathcal{F}, \mathbb{P})$. For any constant $a > 0$, $\mathbb{P}(X \ge a) \le \frac{\mathbb{E}[X]}{a}$.
\end{proposition}
\section{Quantitative Verification of Finite-time Constrained Occupation Time}
\label{sec:main_results}
We derive sufficient barrier conditions to solve Problem \ref{problem}, consisting of one class for upper-bounding the constrained occupation probability of interest and two distinct classes for lower-bounding it. 

The construction of these barrier conditions relies on a stopped process.

\begin{definition}[Stopped Process]
\label{sp}
Let $\tau_{safe} = \inf \{ t > 0 : X_t \notin \mathcal{X} \}$ be the safety exit time. We define the \textit{stopped process} $\{\tilde{X}_t\}_{t \geq 0}$ as the trajectory of the system \eqref{eq:sde} frozen at the moment it touches the boundary of the safe set:
\begin{equation}
    \tilde{X}_t = X_{t \wedge \tau_{safe}} = 
    \begin{cases} 
        X_t & \text{if } t < \tau_{safe}, \\
        X_{\tau_{safe}} & \text{if } t \geq \tau_{safe}.
    \end{cases}
\end{equation}
By construction, the stopped process $\tilde{X}_t$ remains within the closure $\overline{\mathcal{X}}$ for all $t \geq 0$. This formulation ensures that any functional of the trajectory (such as the occupation time or barrier value) ceases to evolve upon safety violation.
\end{definition}

The stopped process $\tilde{X}_t$ inherits right-continuity and the strong Markov property from $X_t$. It evolves according to the original SDE \eqref{eq:sde} on $\mathcal{X}$, and becomes constant after reaching the boundary $\partial \mathcal{X}$. Let $\widetilde{\mathcal L}$ denote the infinitesimal generator of $\tilde X_t$, defined by $\widetilde{\mathcal L}v(x)
=
\lim_{t\to 0^+}
\frac{\mathbb{E}[v(\tilde X_t)\mid \tilde X_0=x] - v(x)}{t}$. Then for any $v \in C^2(\overline{\mathcal{X}})$, we have
\begin{equation}
\widetilde{\mathcal L}v(x)=
\begin{cases}
\mathcal L v(x), & x \in \mathcal{X},\\[4pt]
0, & x \in \partial \mathcal{X},
\end{cases}
\end{equation}
where the generator $\mathcal L$ of the SDE \eqref{eq:sde} is given by \eqref{eq:generator}\cite{kushner1967}.

The following lemma establishes that the occupation time computed on the stopped trajectory is identical to the constrained occupation time of the original system, allowing us to perform verification directly on the stopped dynamics.

\begin{lemma}[Equivalence of Occupation Times]
\label{lemma:eot}
Let $\tilde{X}_t = X_{t \wedge \tau_{safe}}$ be the stopped process, and let $\tilde{O}_{\mathcal{T}}(t) = \int_0^t 1_{\mathcal{T}}(\tilde{X}_s) ds$ be the occupation time of the stopped process.
Then, for any horizon $H \in (0, \infty)$:
\begin{equation*}
    O_{\mathcal{T}}(H) = \tilde{O}_{\mathcal{T}}(H) \quad \text{almost surely.}
\end{equation*}
Consequently, the verification problem is equivalent:
\begin{equation*}
    \mathbb{P}(O_{\mathcal{T}}(H) \geq K) = \mathbb{P}(\tilde{O}_{\mathcal{T}}(H) \geq K).
\end{equation*}
\end{lemma}

\begin{proof}
We consider the behavior of the integrals for any sample path $\omega$:
\begin{enumerate}
    \item $\tau_{safe}(\omega) \geq H$:
    If the system remains safe throughout the horizon $[0,H]$, then $t \wedge \tau_{safe} = t$ for all $t \in [0, H]$. The stopped process $\tilde{X}_t$ coincides exactly with $X_t$, so the integrals are identical.
    \item $\tau_{safe}(\omega) < H$:
    For $t \leq \tau_{safe}$, $\tilde{X}_t = X_t$, so the accumulation is identical.
    For $t > \tau_{safe}$, the stopped process satisfies $\tilde{X}_t = X_{\tau_{safe}} \in \partial \mathcal{X}$. Since the target set $\mathcal{T}$ is strictly contained in the open safe set $\mathcal{X}$, it is disjoint from the boundary ($\mathcal{T} \cap \partial \mathcal{X} = \emptyset$). Consequently, $1_{\mathcal{T}}(\tilde{X}_t) = 0$ for all $t > \tau_{safe}$.
    Thus, $\tilde{O}_{\mathcal{T}}(H)$ stops accumulating exactly at $\tau_{safe}$, matching the definition of $O_{\mathcal{T}}(H)$ in \eqref{opt}.
\end{enumerate}
The proof is completed.
\end{proof}

\subsection{Upper Bounds via Dissipative Barriers}

In this subsection, we introduce \textit{dissipative barriers} to establish upper bounds on constrained occupation probability.

\noindent \textbf{Intuition for Theorem \ref{thm:sde_upper}:} 
Theorem \ref{thm:sde_upper} below models the barrier function $v(x)$ through a drift-based competition mechanism. Inside the target set $\mathcal{T}$, the generator condition imposes a dissipative constraint that offsets the growth induced by the exponential occupation-time weight. In particular, the inequality $\mathcal{L}v(x)+\lambda v(x)\le\beta$ ensures that the drift of the weighted process remains controlled even when the trajectory spends time inside $\mathcal{T}$. The constant $\beta$ represents the maximum allowable upward bias in the barrier dynamics across the state space. To formalize this mechanism, we analyze the time-weighted process  $Z_t = e^{\lambda \tilde{O}_{\mathcal{T}}(t)} v(\tilde{X}_t)$.  Rather than requiring $Z_t$ to be a strict supermartingale, the drift condition allows a bounded growth rate governed by $\beta$. By bounding the expectation $\mathbb{E}[Z_\tau\mid \tilde X_0=x_0]$ up to the horizon $H$ and applying Markov's inequality, we derive a tail bound on the probability that the cumulative occupation time exceeds a threshold $K$:  $\mathbb{P}\bigl(\tilde{O}_{\mathcal{T}}(H) \ge K\bigr) 
\le
e^{-\lambda K}
\left(
v(x_0)
+
\frac{\beta}{\lambda}
\bigl(e^{\lambda H}-1\bigr)
\right)$.  
This bound captures both the initial barrier value $v(x_0)$ and the accumulated drift over the time horizon, each attenuated by the exponential factor $e^{-\lambda K}$. The condition $v(x) \ge 1$ on $\mathcal{T}$ ensures that when the trajectory spends time inside the target set, the exponential weight contributes at least $e^{\lambda O_{\mathcal{T}}(t)}$. In particular, on the event that the occupation time reaches $K$, the weighted process is bounded below by $e^{\lambda K}$, enabling the expectation bound to yield a nontrivial probability bound.

\begin{theorem}[Finite-Horizon Upper Bounds] \label{thm:sde_upper}
Let $v \in C^2(\overline{\mathcal{X}})$ be a non-negative function. Let $\lambda > 0$ be the decay rate and $\beta \geq 0$ be the additive drift constant. Suppose $v$ satisfies:
\begin{enumerate}
    \item \textbf{Dissipative Drift Condition:}
    \begin{equation}
        \mathcal{L}v(x) \leq 
        \begin{cases} 
            -\lambda v(x) + \beta & \text{if } x\in \mathcal{T}, \\
            \beta & \text{if } x \in \mathcal{X}\setminus \mathcal{T}.
        \end{cases}
    \end{equation}
    \item \textbf{Target Positivity:} $v(x) \geq 1$ for all $x\in \mathcal{T}$.
      \item \textbf{Sink Condition:} $-\lambda v(x) + \beta \geq 0$ for all $x \in \partial \mathcal{X}$.
\end{enumerate}
Then, for any finite time horizon $H > 0$ and occupation threshold $K > 0$ with $K\leq H$, the probability of the constrained occupation time exceeding $K$ is bounded by:
\begin{equation*}
    \mathbb{P}(O_{\mathcal{T}}(H) \geq K) \leq e^{-\lambda K} \left( v(x_0) + \frac{\beta}{\lambda} (e^{\lambda H} - 1) \right).
\end{equation*}
\end{theorem}
\begin{proof}
    The proof is shown in Appendix.
\end{proof}

\subsection{Lower Bounds via Attractive Barriers}

In this subsection, we introduce two \textit{Attractive barriers} to lower bound the constrained occupation  probability that the system accumulates at least $K$ time units.

\paragraph{Attractive Barriers I}

\noindent \textbf{Intuition for Theorem~\ref{thm:sde_lower_global}.}
Theorem \ref{thm:sde_lower_global} establishes a lower bound on the constrained occupation probability by constructing the scorekeeping process $Z_t = v(\tilde{X}_t)e^{-\lambda I_{\text{out}}(t)}$,  which combines a barrier function $v$ with an exponential penalty for time spent outside the target. Here, $I_{out}(t) = \int_0^t 1_{\overline{\mathcal{X}}\setminus \mathcal{T}}(\tilde{X}_s) ds$ measures the time spent outside the target. The theorem imposes location-dependent drift conditions on $v$ that regulate the evolution of $Z_t$ through a controlled leakage rate. Inside the target $\mathcal{T}$ the generator of $v$ is bounded below by a constant $\beta \le 0$, allowing the score to decrease at most at a constant rate that represents admissible leakage. Outside $\mathcal{X}\setminus\mathcal{T}$ the drift must satisfy the stronger lower bound $\mathcal{L} v \ge \lambda v + \beta$, so that the generator-induced drift of $v(\tilde{X}_t)$ is sufficiently large to offset the decay caused by the exponential penalty, keeping the overall drift of $Z_t$ bounded below by $\beta e^{-\lambda I_{\mathrm{out}}(t)}$. On the boundary $\partial\mathcal X$ the sink condition $\lambda v+\beta\le0$ guarantees the frozen process does not violate the leakage budget. Taking expectations (the stochastic integral is a true martingale according to Proposition \ref{pro:martin}) yields that the expected score at any stopping time cannot fall below the initial value minus the maximal leakage $|\beta|H$. Finally, splitting the expectation according to success/failure events and using simple upper bounds on the score in each case (success score $\le 1$, failure score $\le \delta$) produces the claimed lower bound on the success probability after algebraic rearrangement.

\begin{theorem}[Attractive Barriers I] \label{thm:sde_lower_global}
Let $v \in C^2(\overline{\mathcal{X}})$ be bounded by $M$. Let $\lambda > 0$ and $\beta \leq 0$. Suppose $v$ satisfies:
\begin{enumerate}
    \item \textbf{Attractive Drift Condition:}
    \begin{equation}
        \mathcal{L}v(x) \geq 
        \begin{cases} 
            \beta & \text{if } x\in \mathcal{T}, \\
            \lambda v(x) + \beta & \text{if } x \in \mathcal{X}\setminus \mathcal{T}.
        \end{cases}
    \end{equation}
    \item \textbf{Bound on Target:} $v(x) \leq 1$ for all $x\in \mathcal{T}$.
    \item \textbf{Sink Condition:} $\lambda v(x) + \beta \leq 0$ for all $x \in \partial \mathcal{X}$.
\end{enumerate}
Then, for any finite horizon $H$ and threshold $K \leq  H$, the probability of success is bounded by:
\begin{equation*}
    \mathbb{P}(O_{\mathcal{T}}(H) \geq K) \geq \frac{v(x_0) - |\beta|H - \delta(H, K)}{1 - \delta(H, K)},
\end{equation*}
where $\delta(H, K) = M e^{-\lambda (H - K)}$, provided $1 - \delta(H, K)>0$.
\end{theorem}
\begin{proof}
    The proof is shown in Appendix.
\end{proof}

In Theorem \ref{thm:sde_lower_global}, we restrict $\beta \leq 0$, as a strictly positive drift ($\beta > 0$) is structurally inadmissible in this framework. Such a positive drift would imply that the expected value of the barrier increases indefinitely ($\mathbb{E}[v(\tilde{X}_t)] \to \infty$), eventually exceeding any finite bound, which creates a mathematical contradiction with the boundedness of $v(x)$ on the compact set $\overline{\mathcal{X}}$ (since the stopped process cannot escape $\overline{\mathcal{X}}$).

\noindent\paragraph{Attractive Barriers II}
\noindent\textbf{Intuition for Theorem~\ref{thm:sde_lower_weighted}.} Theorem \ref{thm:sde_lower_weighted} below establishes lower bounds on the constrained occupation probability for systems exhibiting strict attraction dynamics. The proof employs a bidirectional weighting mechanism that balances the behavior inside the target set $\mathcal{T}$ against the restoring effect from the exterior. This is achieved by constructing the process $Z_t = v(\tilde{X}_t) e^{\lambda(2\tilde{O}_{\mathcal{T}}(t) - t)}$, which increases when the system remains inside the target (accumulating reward) and decreases when it moves outside (incurring penalty). The theorem imposes location-dependent drift conditions on the barrier function $v$ so that $Z_t$ admits a uniform lower bound on its drift with offset $\beta$:
\begin{itemize}
    \item Inside the target ($x \in \mathcal{T}$): $\mathcal{L}v(x) \geq -\lambda v(x) + \beta$
    \item Outside the target ($x \in \mathcal{X}\setminus\mathcal{T}$): $\mathcal{L}v(x) \geq \lambda v(x) + \beta$
    \item On the boundary ($x \in \partial\mathcal{X}$): $\lambda v(x) + \beta \leq 0$
\end{itemize}
These conditions ensure that the exponential growth inside the target and the decay outside balance in such a way that the weighted process $Z_t$ maintains a controlled drift. 
    A positive offset $\beta$ captures systems with strong inward drift toward the target. The certificate can then exploit this positive drift contribution to obtain tighter probability bounds (the bound includes $\Gamma$ which is positive and increases with $\beta$). By analyzing the expected value of $Z_t$ up to the stopping time $\tau=\tau_K\wedge H$, the proof derives a rigorous lower bound on the constrained occupation probability.

\begin{theorem}[Attractive Barriers II] \label{thm:sde_lower_weighted}

Let $v \in C^2(\overline{\mathcal{X}})$ be bounded by $M$. Let $\lambda > 0$ be the decay rate and $\beta \geq 0$ be the drift constant. Suppose $v$ satisfies:

\begin{enumerate}
\item \textbf{Attractive Drift Condition:}
\begin{equation}
\mathcal{L}v(x) \geq
\begin{cases}
-\lambda v(x) + \beta & \text{if } x\in \mathcal{T}, \\
\lambda v(x) + \beta & \text{if } x \in \mathcal{X}\setminus \mathcal{T}.
\end{cases}
\end{equation}

\item \textbf{Boundary Condition:} $v(x) \leq 1$ on $\mathcal{T}$.
\item \textbf{Sink Condition}: $\lambda v(x) + \beta \leq 0$ on $\partial \mathcal{X}$.
\end{enumerate}

Then, for any finite horizon $H \geq  K$ satisfying $e^{\lambda K} - \delta_W(H, K)>0$, 
the constrained occupation probability  is bounded by:
\begin{equation*}
\mathbb{P}(O_{\mathcal{T}}(H) \geq K) \geq \frac{v(x_0) + \Gamma(H, K, \beta) - \delta_W(H, K)}{e^{\lambda K} - \delta_W(H, K)}.
\end{equation*}
Here, $\delta_W(H, K) = M e^{\lambda (2K - H)}$ is the horizon penalty and $\Gamma(H, K, \beta)=\frac{\beta}{\lambda}(1-e^{-\lambda K})$ is the drift term.
\end{theorem}
\begin{proof}
    The proof is shown in Appendix.
\end{proof}

Theorem~\ref{thm:sde_lower_global} and Theorem~\ref{thm:sde_lower_weighted} should be viewed as complementary results; neither strictly dominates the other. Both provide rigorous lower bounds for the constrained occupation probability, but they are tailored to different dynamical regimes through their distinct weighting structures. Theorem~\ref{thm:sde_lower_global} is restricted to the offset parameter $\beta \le 0$ and adopts a budget-style accounting of occupation time. The case $\beta=0$ corresponds to a neutral baseline in which the score process does not lose mass in expectation. In contrast, $\beta < 0$ allows for controlled leakage of the certificate’s value over time. This reflects regimes where the system may experience unfavorable drift even within the target—potentially causing the trajectory to leave the set intermittently or simply flow through it—yet can still accumulate the required total duration over the finite horizon (as demonstrated in Example \ref{example2}). In contrast, Theorem~\ref{thm:sde_lower_weighted} employs a bidirectional weighting $e^{\lambda(2\tilde{O}_{\mathcal{T}}(t)-t)}$ that increases while the trajectory remains inside the target and decreases while it stays outside. This allows the framework to explicitly exploit positive drift ($\beta > 0$) that may arise in strongly attractive systems. Thus, Theorem~\ref{thm:sde_lower_weighted} could yield tighter bounds when such attractive dynamics are present, as the bidirectional weighting prevents the certificate value from escaping to infinity even when $\beta$ is positive (illustrated in Example \ref{example1}).

\section{Examples}
\label{sec:examples}
In this section, we demonstrate the effectiveness of the proposed barrier framework through numerical case studies on stochastic polynomial systems. To synthesize the barrier functions $v(x)$ and verify the drift conditions in Theorems \ref{thm:sde_upper}--\ref{thm:sde_lower_weighted}, we use SOS  programming, which converts the sufficient conditions for occupation-time bounds into SDPs. All SDPs are solved using Mosek~\cite{aps2019mosek}. As a benchmark, we estimate the constrained occupation probabilities using $10^5$ Monte Carlo (MC) simulations based on the Euler--Maruyama method with step size $\Delta t = 2 \times 10^{-3}$.

To avoid bilinearity arising from jointly optimizing the parameters $(\lambda, M)$ and the barrier function $v$, we adopt a pragmatic approach. Specifically, we perform a grid search over the scalar decay rate $\lambda$ and the global bound $M$. For each fixed pair $(\lambda, M)$, the synthesis of a suitable barrier function $v(\bm{x})$ reduces to a convex optimization problem. In addition, unless otherwise specified, we use polynomial barrier functions of degree $d$, containing all monomials with total degree less than or equal to $d$.

\begin{example}
\label{example1}
We consider a polynomial SDE on the domain $\mathcal{X} = [-1, 1]$ with multiplicative noise: 
\begin{equation}
    dX_t = (15X_t^3 - 5X_t)dt + X_tdW_t
\end{equation}
The system admits a locally stable equilibrium at the origin. This behavior is visualized in Figure \ref{fig:trajectories1}. We verify the probability of accumulating a service time of $K = 2.0$ within a narrow target region $\mathcal{T} = [-0.1, 0.1]$ over a horizon $H = 10.0$, starting from $x_0 = 0.5$.

We evaluate polynomial barrier certificates of degree $d=12$ across two different decay rates ($\lambda \in \{0.1, 10^{-5}\}$). Table \ref{tab:barrier_comparison} summarizes the computed bounds against the empirical MC probability of $0.7564$.

\noindent\textbf{Discussion of Results (Table \ref{tab:barrier_comparison}):}
\begin{itemize}
    \item \textbf{Sensitivity to Decay Rate:} The choice of $\lambda$ dictates the performance of different barrier formulations. When $\lambda = 10^{-5}$, Theorem \ref{thm:sde_lower_weighted} leverages the strict attraction of the system to yield a lower bound of $0.5054$. Conversely, Theorem 2 fails at this decay rate, yielding a trivial bound of 0.
    \item \textbf{Theorem \ref{thm:sde_lower_global} at High $\lambda$:} When $\lambda = 0.1$, Theorem 2 recovers a bound of $0.3595$, outperforming the variations of Theorem \ref{thm:sde_lower_weighted}, which fall to $\approx 0.159$.
    \item \textbf{Upper Bound Analysis:} The upper bound from Theorem \ref{thm:sde_upper} tightens to $0.7924$ at $\lambda = 10^{-5}$, correctly identifying that stochasticity drives the system out of the safe region in approximately $24\%$ of trials.
\end{itemize}

\begin{figure}[h]
\centering \includegraphics[width=0.3\textwidth]{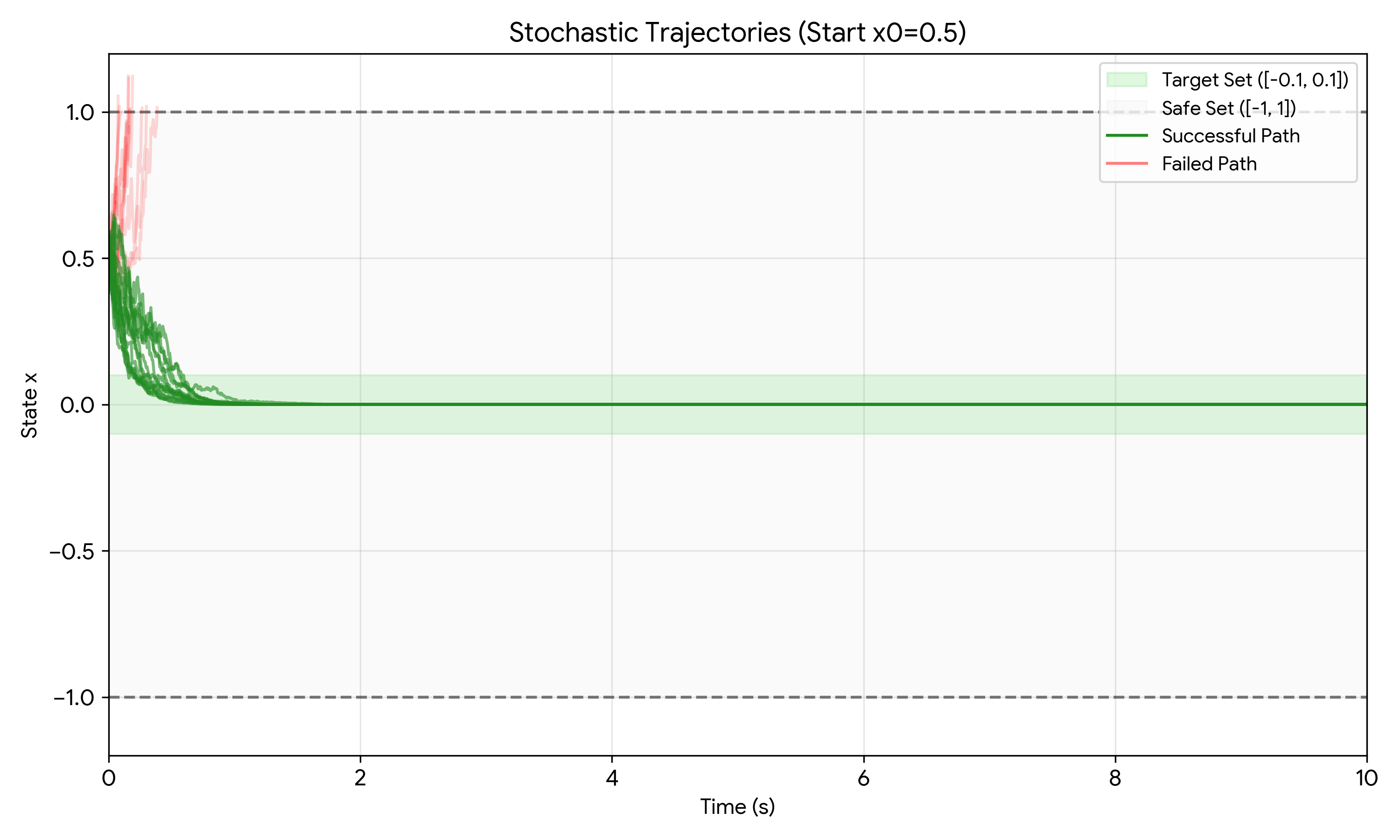} \caption{Sample trajectories ($H=10$) of the system starting from $x_0=0.5$. Green paths successfully accumulate $K=2.0$ in the target region (shaded green) within the horizon $H=10$. Red paths fail by exiting the safe set (shaded gray) prematurely.} \label{fig:trajectories1}
\end{figure}

\begin{table*}[h]
\centering
\caption{Optimization Parameters and Computational Results for Example \ref{example1} ($H=10, K=2$)}
\label{tab:barrier_comparison}

\begin{tabular}{c c c c c c c}
\hline
\textbf{Degree} &\textbf{Bound ($M$)}& $\lambda$& \textbf{Upper Bound} & \textbf{Lower Bound} & \textbf{Lower Bound} \\
&
&
& (Theorem \ref{thm:sde_upper}) 
& (Theorem \ref{thm:sde_lower_global}) 
& (Theorem \ref{thm:sde_lower_weighted})\\ 
\hline

\multirow{2}{*}{12} 
&\multirow{2}{*}{1} 
& 0.1
& 1
& 0.3595 
& 0.1595\\

&
& $10^{-5}$
& 0.7924
& 0
& 0.5054\\

\hline
\textbf{MC} & \multicolumn{5}{c}{\textbf{0.7564}} \\
\hline
\end{tabular}
\end{table*}

\begin{table*}[h]
    \centering
    \caption{Optimization Parameters and Computational Results for Example \ref{example2} ($d=10$)}
    \begin{tabular}{l c c l }
        \hline
        \textbf{Method} & \textbf{$\lambda$} & \textbf{Bound ($M$)} & \textbf{Certified Probability} \\
        \hline
        Theorem \ref{thm:sde_upper} & $0.5$ & --  & $\le \mathbf{0.7162}$ \\
        \hline
        Theorem \ref{thm:sde_lower_global} & $10^{-5}$ & $1$ & $\ge \mathbf{0.6570}$ \\
        \hline
        Theorem \ref{thm:sde_lower_weighted}  & $\{10^{-5}, 10^{-4}, 10^{-3}, 10^{-2}, 10^{-1},1\}$ & $1$ & \textit{Infeasible} \\
        \hline
\textbf{MC} & \multicolumn{3}{c}{\textbf{0.6919}} \\
\hline
    \end{tabular}
    \label{tab:parameters}
\end{table*}

\end{example}

\begin{example}
\label{example2}
  We analyze the system: 
\begin{equation}
    dX_t = (X_t^3 - 5X_t)dt + X_tdW_t.
\end{equation}
Trajectories start at $x_0 = 0.9$, and we bound the probability of accumulating at least $K = 0.1$ seconds in $\mathcal{T} = [0.1, 0.5]$ over $H = 5.0$. 
The behavior is visualized in Figure \ref{fig:trajectories}. We use degree $d=10$ polynomials. Table \ref{tab:parameters} summarizes the computed bounds.

\textbf{Discussion of Results (Table \ref{tab:parameters}):}
\begin{itemize}
    \item \textbf{Tight Bracketing of Probability:} The proposed framework successfully brackets the true probability within a tight interval: $[0.6570, 0.7162]$. The lower bound of $0.6570$ is obtained via Theorem \ref{thm:sde_lower_global} (Attractive Barriers I), while the upper bound of $0.7162$ is derived from Theorem \ref{thm:sde_upper}. This tightly encompasses the MC success probability of $0.6916$.
    \item \textbf{Infeasibility of Theorem \ref{thm:sde_lower_weighted}:} For this problem, the SOS constraints used to implement Theorem \ref{thm:sde_lower_weighted} were infeasible across the tested parameter grids. 
\end{itemize}

\begin{figure}[htbp]
    \centering
    \includegraphics[width=0.3\textwidth]{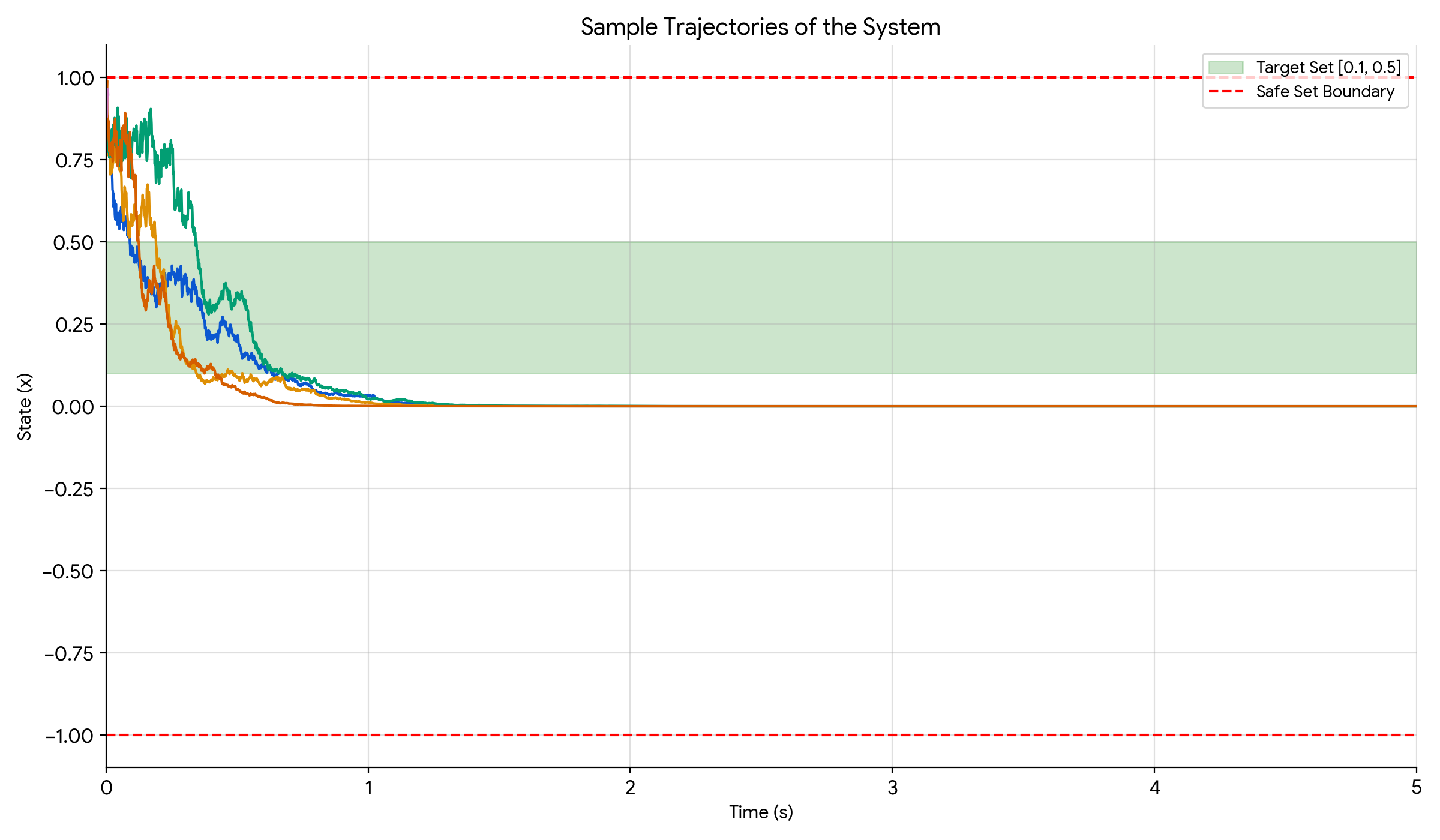}
    \caption{Sample trajectories ($H=5$) of the system starting at \(x_0=0.9\). The paths are driven rapidly towards the origin by the drift and then linger within \(\mathcal{T}=[0.1, 0.5]\) (green shaded region) due to the vanishing diffusion term, illustrating the high probability of satisfying the occupation time constraint.}
    \label{fig:trajectories}
\end{figure}

\end{example}

\section{Conclusion}
\label{sec:conclusion}
This paper presented a unified barrier certificate framework for the quantitative verification of finite-time constrained occupation measures in continuous-time stochastic systems modeled by SDEs. By shifting the verification focus from instantaneous reachability to cumulative service duration, we addressed a critical requirement for autonomous systems operating under uncertainty. We introduced three distinct classes of barrier certificates for bounding constrained occupation probabilities. Numerical validation demonstrated the framework’s efficacy and versatility.

Future work will focus on automating the selection of key parameters (e.g., $\lambda$, $M$), improving the feasibility of the more complex barrier constructions, and analyzing the conservativeness of the proposed conditions both theoretically and through numerical experiments.

\bibliographystyle{splncs04}
\bibliography{ref}

\section*{Appendix}

\paragraph{\textbf{The proof of Theorem \ref{thm:sde_upper}}}

\begin{proof}
We analyze the stopped process $\tilde{X}_t = X_{t \wedge \tau_{safe}}$ and its corresponding occupation time $\tilde{O}_{\mathcal{T}}(t) = \int_0^t 1_{\mathcal{T}}(\tilde{X}_s) ds$. By Lemma \ref{lemma:eot}, $\mathbb{P}(O_{\mathcal{T}}(H) \geq K) = \mathbb{P}(\tilde{O}_{\mathcal{T}}(H) \geq K)$.

\textbf{Step 1: Construction of the Scaled Process.}
We define the time-dependent stochastic process $Z_t$ on the stopped trajectory: $Z_t = e^{\lambda \tilde{O}_{\mathcal{T}}(t)} v(\tilde{X}_t)$.

We differentiate $Z_t$ using It\^{o}'s Product Rule. For $t < \tau_{safe}$, the dynamics follow the original SDE. For $t \geq \tau_{safe}$, the state is frozen at $\partial \mathcal{X}$, implying $dZ_t = 0$. Thus, we analyze the differential up to the safety exit:
\begin{equation*}
    dZ_t = e^{\lambda \tilde{O}_{\mathcal{T}}(t)} \left[ 
    \begin{aligned}
    \left( \mathcal{L}v(\tilde{X}_t) + \lambda 1_{\mathcal{T}}(\tilde{X}_t) v(\tilde{X}_t) \right) dt \\
    + \nabla v(\tilde{X}_t) \sigma(\tilde{X}_t) dW_t 
    \end{aligned}
        \right].
\end{equation*}

\textbf{Step 2: Drift Analysis.}
We analyze the drift term $\mu_Z(t) = e^{\lambda \tilde{O}_{\mathcal{T}}(t)} (\mathcal{L}v(\tilde{X}_t) + \lambda 1_{\mathcal{T}}(\tilde{X}_t) v(\tilde{X}_t))$:
\begin{itemize}
    \item If $\tilde{X}_t \in \mathcal{T}$: $1_{\mathcal{T}}=1$. By condition (1), $\mathcal{L}v \leq -\lambda v + \beta$, so $\mu_Z(t) \leq \beta e^{\lambda \tilde{O}_{\mathcal{T}}(t)}$.
    \item If $\tilde{X}_t \in \mathcal{X}\setminus \mathcal{T}$: $1_{\mathcal{T}}=0$. By condition (1), $\mathcal{L}v \leq \beta$, so $\mu_Z(t) \leq \beta e^{\lambda \tilde{O}_{\mathcal{T}}(t)}$.
    \item If $\tilde{X}_t \in \partial \mathcal{X}$ (stopped): The drift is identically zero, which satisfies $\beta e^{\lambda t}$ since $\beta\geq 0$.
\end{itemize}
In all cases, since $\tilde{O}_{\mathcal{T}}(t) = \int_0^t 1_{\mathcal{T}}(\tilde{X}_s) ds\leq t$, we obtain that $e^{\lambda \tilde{O}_{\mathcal{T}}(t)} \leq e^{\lambda t}$. Further, the drift is upper-bounded by $\beta e^{\lambda \tilde{O}_{\mathcal{T}}(t)} \leq \beta e^{\lambda t}$.

\textbf{Step 3: Integral Form and Expectation.}
Let $\tau_K=\inf\{t:\tilde O_{\mathcal{T}}(t)=K\}$ be the completion time and define the bounded stopping time $\tau=\tau_K\wedge H$.

Integrating the SDE for $Z_t$ up to time $\tau$ gives
\[
Z_\tau
=
v(x_0)
+
\int_0^\tau \mu_Z(s)ds
+
\int_0^\tau
e^{\lambda \tilde O_{\mathcal{T}}(s)}
\nabla v(\tilde X_s)\sigma(\tilde X_s)dW_s .
\]

Since $\nabla v$ and $\sigma$ are bounded on the compact domain $\overline{\mathcal{X}}$, and $\tau \le H$, the stochastic integral term is square-integrable. Therefore, it has zero expectation and is a true martingale according to Proposition \ref{pro:martin}. Taking expectations yields $\mathbb{E}[Z_\tau \mid \tilde X_0=x_0]
=
v(x_0)
+
\mathbb{E}\!\left[\int_0^\tau \mu_Z(s)ds \mid \tilde X_0=x_0\right]$.

Using the drift bound and extending the integral to $H$, we have $\mathbb{E}[Z_\tau\mid \tilde X_0=x_0]
\le
v(x_0)
+
\int_0^H
\beta e^{\lambda s}ds$. Evaluating the integral gives $\mathbb{E}[Z_\tau\mid \tilde X_0=x_0]
\le
v(x_0)
+
\frac{\beta}{\lambda}(e^{\lambda H}-1)$.

\textbf{Step 4: Markov's Inequality.}
Consider the success event $S = \{ \tilde{O}_{\mathcal{T}}(H) \geq K \}$. On this event, $\tau = \tau_K \leq H$.
This implies $\tilde{O}_{\mathcal{T}}(\tau) = K$ and $\tilde{X}_{\tau} \in \mathcal{T}$. By condition (2), $v(\tilde{X}_{\tau}) \geq 1$.
Thus, $Z_{\tau} = e^{\lambda K} v(\tilde{X}_{\tau}) \geq e^{\lambda K}$ on $S$.
Applying Markov's inequality, we have
\begin{equation*}
    \mathbb{P}(S) \leq \frac{\mathbb{E}[Z_{\tau}]}{e^{\lambda K}} \leq e^{-\lambda K} \left( v(x_0) + \frac{\beta}{\lambda}(e^{\lambda H} - 1) \right).
\end{equation*}
The proof is completed.
\end{proof}

\paragraph{\textbf{The proof of Theorem \ref{thm:sde_lower_global}}}

\begin{proof}
We analyze the stopped process $\tilde{X}_t = X_{t \wedge \tau_{safe}}$ and its corresponding occupation time $\tilde{O}_{\mathcal{T}}(t) = \int_0^t 1_{\mathcal{T}}(\tilde{X}_s) ds$. By Lemma \ref{lemma:eot}, $\mathbb{P}(O_{\mathcal{T}}(H) \geq K) = \mathbb{P}(\tilde{O}_{\mathcal{T}}(H) \geq K)$.

 \textbf{Step 1: Construction of the Scaled Stopped Process.}
We define the stochastic process $Z_t$ on the stopped trajectory $\tilde{X}_t = X_{t \wedge \tau_{safe}}$: $Z_t = v(\tilde{X}_t) e^{-\lambda I_{out}(t)}$,
where $I_{out}(t) = \int_0^t 1_{\overline{\mathcal{X}}\setminus \mathcal{T}}(\tilde{X}_s) ds$ measures the time spent outside the target. 

We differentiate $Z_t$ using It\^{o}'s Product Rule. For $t < \tau_{safe}$, the dynamics follow the original SDE. For $t \geq \tau_{safe}$, the state is frozen at $\partial \mathcal{X}$, implying $dZ_t = 0$. Thus, we analyze the differential up to the safety exit:
\begin{equation*}
    dZ_t = e^{-\lambda I_{out}(t)} \left[ 
    \begin{aligned}
    \left( \mathcal{L}v(\tilde{X}_t) - \lambda 1_{\overline{\mathcal{X}}\setminus \mathcal{T}}(\tilde{X}_t) v(\tilde{X}_t) \right) dt \\
    + \nabla v(\tilde{X}_t) \sigma(\tilde{X}_t) dW_t 
    \end{aligned}
        \right].
\end{equation*}

Note that if the system exits the safe set, $\tilde{X}_t$ freezes at the boundary $\partial \mathcal{X}$ (which is disjoint from $\mathcal{T}$), causing $I_{out}(t)$ to grow at rate 1 indefinitely.

\textbf{Step 2: Drift Analysis.}
We analyze the drift $\mu_Z(t)=e^{-\lambda I_{out}(t)}
    \left( \mathcal{L}v(\tilde{X}_t) - \lambda 1_{\overline{\mathcal{X}}\setminus \mathcal{T}}(\tilde{X}_t) v(\tilde{X}_t) \right)$ of the process $Z_t$ over the entire domain $\overline{\mathcal{X}}$. Applying It\^{o}'s formula for the stopped process:
\begin{itemize}
    \item \textbf{Case 1 (Interior, Inside Target $x\in \mathcal{T}$):} 
    Here $1_{\mathcal{X}\setminus \mathcal{T}} = 0$. The dynamics follow the original SDE. Thus, $\mu_Z(t) = e^{-\lambda I_{out}(t)} \mathcal{L}v(\tilde{X}_t)$. By the Attractive Drift Condition, $\mathcal{L}v(x) \geq \beta$. Thus, $\mu_Z(t) \geq \beta e^{-\lambda I_{out}(t)}$.

    \item \textbf{Case 2 (Interior, Outside Target $x \in \mathcal{X}\setminus \mathcal{T}$):} 
    Here $1_{\mathcal{X}\setminus \mathcal{T}} = 1$. The dynamics follow the original SDE, i.e., $\mu_Z(t) = e^{-\lambda I_{out}(t)} (\mathcal{L}v(\tilde{X}_t) - \lambda v(\tilde{X}_t))$.     By the condition $\mathcal{L}v(x) \geq \lambda v(x) + \beta$, we have $\mathcal{L}v(x) - \lambda v(x) \geq \beta$. Thus, $\mu_Z(t) \geq \beta e^{-\lambda I_{out}(t)}$.

    \item \textbf{Case 3 (Boundary $x \in \partial \mathcal{X}$):} 
    Here $\tilde{X}_t$ is frozen, so $d\tilde{X}_t = 0$, implying the generator term vanishes ($\mathcal{L}v(\tilde{X}_t) = 0$). However, the time counter continues to grow ($dI_{out} = 1$). The differential is driven solely by the exponential decay:
    \begin{equation*}
        dZ_t = v(\tilde{X}_t) \left( -\lambda e^{-\lambda I_{out}(t)} dt \right) = -\lambda v(\tilde{X}_t) e^{-\lambda I_{out}(t)} dt.
    \end{equation*}
    The drift is $\mu_Z = -\lambda v(\tilde{X}_t) e^{-\lambda I_{out}(t)}$. To maintain the global lower bound $\mu_Z \geq \beta e^{-\lambda I_{out}(t)}$, we require $-\lambda v(x) \geq \beta$, which rearranges to $\lambda v(x) + \beta \leq 0$.
    This is exactly the \textbf{Sink Condition}.
\end{itemize}
Thus, $\mu_Z(t) \geq \beta e^{-\lambda I_{out}(t)}$ holds globally for all $t \geq 0$.

\textbf{Step 3: Integration.}
Let $\tau_K=\inf\{t:\tilde O_{\mathcal{T}}(t)=K\}$ be the completion time and define the bounded stopping time $\tau=\tau_K\wedge H$.

 We integrate the SDE from $0$ to $\tau$. Since $\nabla v$ and $\sigma$ are bounded on the compact domain $\overline{\mathcal{X}}$, and $\tau \le H$, the stochastic integral term is square-integrable. Therefore, it has zero expectation and is a true martingale according to Proposition \ref{pro:martin}. Taking expectations yields $\mathbb{E}[Z_{\tau}\mid \tilde X_0=x_0] = v(x_0) + \mathbb{E}\left[ \int_0^{\tau} \mu_Z(s) ds \mid  \tilde X_0=x_0\right] \geq v(x_0) + \mathbb{E}\left[ \int_0^{\tau} \beta e^{-\lambda I_{out}(s)} ds \mid \tilde X_0=x_0\right]$. To obtain a valid lower bound given $\beta \leq 0$, we maximize the magnitude of the subtracted integral. Using $I_{out}(s) \geq 0$ and extending the limit to $H$:
\begin{equation} 
\label{eq:lower_bound_exp_final}
    \mathbb{E}[Z_{\tau}\mid \tilde X_0=x_0] \geq v(x_0) + \int_0^\tau \beta ds \geq v(x_0) -|\beta|H.
\end{equation}

\textbf{Step 4: Outcome Decomposition.}
We decompose the outcome space into Success ($S = \{\tilde{O}_{\mathcal{T}}(H) \geq K\}$) and No Success ($S^c$).
\begin{enumerate}
    \item \textbf{On Success ($S$):} $\tau = \tau_K \leq H$. Thus $\tilde{O}_{\mathcal{T}}(\tau)=K$ and $\tilde{X}_\tau \in \mathcal{T}$. By the target bound, $v \leq 1$, so $Z_\tau \leq 1$.
    \item \textbf{On No Success ($S^c$):} The horizon $H$ is reached with $\tilde{O}_{\mathcal{T}}(H) < K$.
    Whether the system is in $\mathcal{X}\setminus \mathcal{T}$ or frozen at the boundary $\partial \mathcal{X}$, the global bound $|v| \leq M$ applies. The time spent outside is $I_{out} = H - \tilde{O}_{\mathcal{T}} > H - K$.
    Thus, $Z_\tau \leq M e^{-\lambda (H - K)} = \delta(H,K)$.
\end{enumerate}
Combining these, we have $\mathbb{E}[Z_{\tau} \mid \tilde X_0=x_0] \leq \mathbb{P}(S) \cdot 1 + (1 - \mathbb{P}(S)) \cdot \delta = \mathbb{P}(S)(1 - \delta) + \delta(H,K)$. Substituting the lower bound from \eqref{eq:lower_bound_exp_final}, we obtain $\mathbb{P}(S)(1 - \delta(H,K)) + \delta(H,K) \geq v(x_0) - |\beta|H$. Rearranging for $\mathbb{P}(S)$ yields the  result.
\end{proof}

\paragraph{\textbf{The proof of Theorem \ref{thm:sde_lower_weighted}}}

\begin{proof}
We analyze the stopped process $\tilde{X}_t = X_{t \wedge \tau_{safe}}$ and its corresponding occupation time $\tilde{O}_{\mathcal{T}}(t) = \int_0^t 1_{\mathcal{T}}(\tilde{X}_s) ds$. By Lemma \ref{lemma:eot}, $\mathbb{P}(O_{\mathcal{T}}(H) \geq K) = \mathbb{P}(\tilde{O}_{\mathcal{T}}(H) \geq K)$. We analyze the stopped process $\tilde{X}_t$ using the function $Y_t = \lambda(2 \tilde{O}_{\mathcal{T}}(t) - t)$.

\textbf{Step 1: Construction of the Scaled Stopped Process.}  We construct the scaled stochastic process on the stopped trajectory $Z_t = v(\tilde{X}_t)e^{Y_t}$.

To analyze the behavior of $Z_t$, we apply It\^{o}'s product rule. Since $Y_t$ is continuous and of finite variation (its derivative exists almost everywhere as $\dot{Y}_t = \lambda(2\cdot 1_{\mathcal{T}}(\tilde{X}_t) - 1)$), it has zero quadratic variation. Therefore, the differential of $Z_t$ is
\begin{equation*}
dZ_t = e^{Y_t}dv(\tilde{X}_t) + v(\tilde{X}_t)d(e^{Y_t}).
\end{equation*}
For the first term, applying It\^{o}'s formula to $v(\tilde{X}_t)$ yields
$dv(\tilde{X}_t)=\mathcal{L}v(\tilde{X}_t)dt+\nabla v(\tilde{X}_t)\sigma(\tilde{X}_t)dW_t$.
Next, since $Y_t$ is a finite-variation process, $d(e^{Y_t}) = e^{Y_t} dY_t$. Using the definition of $Y_t$, $dY_t = \lambda(2\,1_{\mathcal T}(\tilde X_t)-1)dt$.
Substituting these expressions into the product rule gives $dZ_t =e^{Y_t}\Big(\mathcal{L}v(\tilde{X}_t)+\lambda(2\,1_{\mathcal T}(\tilde{X}_t)-1)v(\tilde{X}_t)\Big)dt+e^{Y_t}\nabla v(\tilde{X}_t)\sigma(\tilde{X}_t)dW_t$. Thus, the drift of $Z_t$ is $\mu_Z(t)=e^{Y_t}\Big(\mathcal{L}v(\tilde{X}_t)+\lambda(2\,1_{\mathcal T}(\tilde{X}_t)-1)v(\tilde{X}_t)\Big)$.

We evaluate this drift in three distinct regimes based on the location of the stopped state $\tilde{X}_t$:

\begin{itemize}
\item \textbf{Case 1 (Inside Target $x \in \mathcal{T}$):} The weighting function grows with $\frac{d Y}{dt} = \lambda(2\cdot 1-1) = \lambda$. Substituting the weighted drift condition $\mathcal{L}v(x) \geq -\lambda v(x) + \beta$, we obtain 
$\mu_Z(t)= e^{Y_t} (\mathcal{L}v(x) + \lambda v(x)) \geq e^{Y_t} (-\lambda v(x) + \beta + \lambda v(x)) = \beta e^{Y_t}$.

\item \textbf{Case 2 (Outside Target $x \in \mathcal{X}\setminus\mathcal{T}$):} The weighting function decays with $\frac{d Y}{dt} = \lambda(2\cdot 0-1) = -\lambda$. Substituting the condition $\mathcal{L}v(x) \geq \lambda v(x) + \beta$, we obtain $\mu_Z(t)= e^{Y_t} (\mathcal{L}v(x) - \lambda v(x)) \geq e^{Y_t} (\lambda v(x) + \beta - \lambda v(x)) = \beta e^{Y_t}$.

\item \textbf{Case 3 (Boundary $x \in \partial\mathcal{X}$):} The system is frozen ($d\tilde{X}_t = 0$), so $\mathcal{L}v(x) = 0$. However, time continues to pass, and since $\partial\mathcal{X} \cap \mathcal{T} = \emptyset$, the weight decays with $\dot{Y}_t = -\lambda$. The drift becomes $\mu_Z(t) = e^{Y_t} (0 - \lambda v(x))$. By the Sink Condition, $\lambda v(x) + \beta \leq 0 \implies -\lambda v(x) \geq \beta$. Thus, $\mu_Z(t) \geq \beta e^{Y_t}$.
\end{itemize}

In all cases, the state-dependent terms cancel perfectly, yielding the global lower bound on the drift: $\mu_Z(t) \geq \beta e^{Y_t}$.

\textbf{Step 2: Integration.}
Let $\tau_K=\inf\{t:\tilde O_{\mathcal{T}}(t)=K\}$ be the completion time and define the bounded stopping time $\tau=\tau_K\wedge H$.

Integrating up to the stopping time $\tau$, we have $
    \mathbb{E}[Z_{\tau}\mid \tilde X_0=x_0] \geq v(x_0) + \beta \cdot \mathbb{E}\left[ \int_0^\tau e^{Y_s} ds \mid \tilde X_0=x_0 \right]$.
Since $Y_s=\lambda(2\tilde O_{\mathcal T}(s)-s)\ge -\lambda s$, we obtain the pathwise bound $\int_0^\tau e^{Y_s}ds
\ge
\int_0^K e^{-\lambda s}ds
=
\frac{1-e^{-\lambda K}}{\lambda}$. Therefore, $
\mathbb E\!\left[\int_0^\tau e^{Y_s}ds \mid \tilde X_0=x_0\right]
\ge \frac{1-e^{-\lambda K}}{\lambda}$. Substituting into the drift inequality gives $\mathbb{E}[Z_\tau\mid \tilde X_0=x_0]
\ge
v(x_0)
+
\beta \frac{1-e^{-\lambda K}}{\lambda}$.
    

\textbf{Step 3: Probability Bound.} We relate the expected terminal value to the probability of success by decomposing the outcome space into two events: Success ($S = {\tilde{O}_{\mathcal{T}}(H) \geq K}$) and No Success ($S^c$).

\begin{enumerate}
\item \textbf{On Success ($S$):} The process stops at $\tau = \tau_K$. At this instant, the occupation time is exactly $K$, and the total time elapsed is $\tau \geq K$. The weighting exponent is $Y_\tau = \lambda(2K - \tau)$. Since $\tau \geq K$, the exponent satisfies $Y_\tau \leq \lambda K$. Combined with the bound $v(x) \leq 1$ on the target set, we have $Z_\tau = v(\tilde{X}_\tau) e^{Y_\tau} \leq 1 \cdot e^{\lambda K} = e^{\lambda K}$.

\item \textbf{On No Success ($S^c$):} The process stops at $\tau = H$ without achieving the occupation threshold ($\tilde{O}_{\mathcal{T}}(H) < K$). The occupation time is effectively bounded by $K$ (since we failed). The weighting exponent is $Y_H = \lambda(2\tilde{O}_{\mathcal{T}}(H) - H) < \lambda(2K - H)$. Using the global bound $|v(x)| \leq M$, the terminal value is bounded by $Z_H \leq M e^{\lambda(2K - H)} = \delta_W(H, K)$. 
\end{enumerate}

Combining these cases, the expected value is upper-bounded by $
\mathbb{E}[Z_\tau\mid \tilde X_0=x_0] \leq \mathbb{P}(S) e^{\lambda K} + (1 - \mathbb{P}(S)) \delta_W(H, K)$.

Finally, we combine this upper bound with the submartingale lower bound derived in Step 2  to obtain the final bound stated in the theorem.
\end{proof}

\end{document}